\documentclass{article}
\usepackage[utf8]{inputenc}
\usepackage{amssymb,amsmath,amsthm}
\usepackage{natbib}
\usepackage{graphicx}
\usepackage{tikz}
\usepackage{algorithm}
\usepackage{algorithmic}

\newtheorem{theorem}{Theorem}
\newtheorem{proposition}{Proposition}

\newtheorem{lemma}{Lemma}

\newcommand{\dpst}{\displaystyle}

\newcommand{\Tau}{\mathrm{T}}
\newcommand{\E}{\mathbb{E}}
\newcommand{\N}{\mathbb{N}}
\newcommand{\R}{\mathbb{R}}
\newcommand{\argsinh}{\mathrm{argsinh}\,}

\newcommand{\bh}{\boldsymbol{h}}
\newcommand{\bt}{\boldsymbol{t}}
\newcommand{\bu}{\boldsymbol{u}}
\newcommand{\bx}{\boldsymbol{x}}
\newcommand{\by}{\boldsymbol{y}}
\newcommand{\bI}{\boldsymbol{I}}

\newcommand{\bzero}{\boldsymbol{0}}
\newcommand{\bomega}{\boldsymbol{\omega}}
\newcommand{\bOmega}{\boldsymbol{\Omega}}

\newcommand{\btau}{\boldsymbol{\tau}}
\newcommand{\bTau}{\boldsymbol{\Tau}}
\newcommand{\absu}{\lvert u \rvert}
\newcommand{\absx}{\lvert x \rvert}

\newcommand{\sm}{{\cal X}}

\title{Simulating space-time random fields with nonseparable Gneiting-type covariance functions}

\author{Denis Allard$^{a,*}$, Xavier Emery$^{b,c}$, Céline Lacaux$^d$, Christian Lantuéjoul$^e$\\
(authors are in alphabetical order)\\
\small{$^a$Biostatistics and Spatial Processes (BioSP), INRA PACA, 84914 Avignon Cedex, France}\\
\small{$^*$ Corresponding author. Email address: denis.allard@inra.fr}\\
\small{$^b$Department of Mining Engineering, University of Chile, Santiago, Chile} \\
\small{$^c$Advanced Mining Technology Center, University of Chile, Santiago, Chile}\\
\small{$^d$Avignon Université, LMA EA 2151, 84000, Avignon, France}\\
\small{$^e$Centre de Géosciences, MINES ParisTech, PSL University, Paris, France}
}

\begin{document}
\medskip
\date{\today}
\maketitle

\section*{Abstract}
Two algorithms are proposed to simulate space-time Gaussian random fields with a covariance function belonging to an extended Gneiting class, the definition of which depends on a completely monotone function associated with the spatial structure and a conditionally negative definite function associated with the temporal structure. In both cases, the simulated random field is constructed as a weighted sum of cosine waves, with a Gaussian spatial frequency vector and a uniform phase. The difference lies in the way to handle the temporal component. The first algorithm relies on a spectral decomposition in order to simulate a temporal frequency conditional upon the spatial one, while in the second algorithm the temporal frequency is replaced by an intrinsic random field whose variogram is proportional to the conditionally negative definite function associated with the temporal structure. Both algorithms are scalable as their computational cost is proportional to the number of space-time locations, which may be unevenly spaced in space and/or in time. They are illustrated and validated through synthetic examples.
\bigskip

\noindent \emph{Keywords:} continuous spectral simulation; spectral measure; substitution random field; Gaussian random field.

\bigskip

%
\section{Introduction}
\label{sec:intro}
%

The simulation of random fields plays an increasingly important role in environmental and climate studies, for example to quantify uncertainties and to  assess adaptation scenarios to global changes. Space-time simulations that span over relatively large regions and long periods of time generate very large space-time grids as soon as the resolution is not coarse. However, simulating space-time random fields on very large grids remains a challenge, in particular for nonseparable covariance functions able to capture space-time complexity, such as the Gneiting class of covariance functions \citep{Gneiting2002}. This motivation prompted this research.

In the following, $Z(\bx,t)$ will denote a space-time random field defined over $\R^k \times \R$, where $k$ is the space dimension, with $k=2$ or $k=3$  in most applications. Here, and in the rest of this work, we will use roman letters for scalars and bold letters for vectors. Without loss of generality, we shall assume that the random field is centered, i.e. $\mathbb{E}[Z(\bx,t)]=0,\ \forall (\bx,t) \in \R^k \times \R$.  It will also be assumed that the random field is second-order stationary, so that the covariance function only depends on the space-time lag $(\bh,u) \in \R^k \times \R$:
\begin{equation}
    \hbox{Cov}( Z(\bx,t) , Z(\bx+\bh,t+u)) = C(\bh,u).
    \label{eq:sta2}
\end{equation}
The functions $C(\bh,0)$ and $C(\bzero,u)$ are purely spatial and temporal covariance functions, respectively. It is well-known that $C$ must be a positive semi-definite function over $\R^k \times \R$, see for example \cite{gneiting2010continuous} and references therein. 

A space-time covariance $C$ is separable if it can be factored into spatial and temporal covariance functions, so that 
\begin{equation}
    C(\bh,u)  = C(\bzero,0)^{-1} C(\bh,0) C(\bzero,u).
    \label{eq:sep}
\end{equation}
Although congenial from a mathematical point of view, separability is often a simplistic assumption in many applications, 
which makes separable covariances unable to take into account sophisticated interactions between space and time.  Nonseparable space-time covariance functions can be constructed from basic building-blocks of purely spatial and/or temporal covariance functions, taking advantage of the fact that the class of covariance functions is closed under products, convex mixtures and limits. The product-sum \citep{de2001space} and  convex mixture \citep{ma2002spatio, ma2003families} approaches 
are two such examples. Another straightforward construction is to consider a space-time geometric anisotropy on $\R^{k+1}$, based on an isotropic covariance model and a geometrical transformation (rotation and rescaling) of the coordinates. 
All these models can easily be simulated, even on very large space-time grids,  by combining well-established purely spatial and purely temporal simulation algorithms, see for example \cite{schlather2015analysis}.

\cite{Gneiting2002} proposed a class of fully symmetric space-time  covariances that has  become one of the standard classes of models for space-time random fields in applications relating to climate variables. It was later extended by \cite{zastavnyi2011characterization}. We refer to this extended class as the Gneiting class or the class of Gneiting-type covariance functions. To the best of our knowledge, specific simulation methods for space-time Gaussian random fields with Gneiting-type covariances have not been developed yet. General methods can be used, but they are of limited applicability. The covariance matrix decomposition \citep{Davis1987} is an exact method, but it is well-known that the computational cost to simulate a random field at $n$ space-time locations $\{(\bx_1,t_1),\dots,(\bx_n,t_n)\}$ is ${\cal O}(n^3 )$. Its application is therefore limited to a few thousands space-time locations only. One way to reduce the computational burden consists in using block circulant matrices 
and FFT, as introduced by \cite{wood1994simulation}, \cite{Dietrich97} and \cite{chan1999simulation}. 

Continuous spectral and turning bands methods, see for example \cite{Shinozuka1971}, \cite{mtn73} and  \cite{Chiles2012} and references therein, are 
known to be computationally efficient to simulate spatial random fields. Moreover, they are scalable, in the sense that after some initial calculations, the computational burden is proportional to the number of locations, which do not need to be regularly located. Finally, most, if not all, spatial covariance models with explicitly known spectral measures can be simulated. For all these reasons, it is tempting to see whether these methods can be extended to work on nonseparable space-time random fields. 

In this work, two simulation methods inspired by continuous spectral approaches are presented to simulate space-time Gaussian random fields with extended Gneiting covariance functions. The outline is as follows. Section~\ref{sec:theory} provides the mathematical background required in the subsequent sections and presents the original and extended Gneiting classes of covariances. Section~\ref{sec:spectral} proposes a continuous spectral approach based on a conditional decomposition  of the spectral measure of $C(\bh,u)$. 
Section~\ref{sec:substitution} proposes an alternative approach that can be seen as a particular case of the substitution approach presented in \cite{Lantu2002}. Both approaches are illustrated with synthetic examples. Their pros and cons are compared and discussed in Section~\ref{sec:discussion}. Finally, Section~\ref{sec:conclusions} presents conclusions and proposes some perspectives to generalize our approaches.

%
\section{Theoretical background}
\label{sec:theory}
%

%
\subsection{Completely monotone and Bernstein functions}
A function on the positive half-line $\varphi(t), t \geq 0$, is said to be completely monotone if it possesses a derivative $\varphi^{(n)}$ for any order $n \in \N$ with $(-1)^n \varphi^{(n)}(t) \geq 0$ for any $t>0$. By Bernstein's theorem \citep{Bernstein1929, Feller1966}, a continuous completely monotone function can be written as the Laplace transform of a nonnegative measure $\mu$, i.e.
\begin{equation}
\varphi(t) = \int_{\R^+}e^{-rt} \mu (dr).
\label{eq:Laplace}
\end{equation}
Furthermore, by Schoenberg's theorem \citep{Schoenberg1938a}, the radial function
\begin{equation}
    \phi(\bh) = \varphi(\lvert\bh\rvert^2), \qquad \bh \in \R^k,
\end{equation}
where $\lvert \bx \rvert = \langle \bx, \bx \rangle^{1/2}$ denotes the Euclidean norm of vector $\bx$ and  $\langle \bx, \by \rangle $ is the usual scalar product between vectors $\bx$ and $\by$, is a covariance function for any dimension $k \in \N^*$ if and only if $\varphi(t), t \geq 0,$ is a completely monotone function.

Closely related to the completely monotone functions are the Bernstein functions. A function $\psi$ 
defined on $\R_+$ is a Bernstein function if it is a positive primitive of a completely monotone function. It can be shown 
\citep{slgsogvdk10} that $\psi$ admits the general expression 
\begin{equation}\label{eq:btn}
\psi (t) =  a + b t + \int_0^{+\infty} \bigl( 1 - e^{ - t x } \bigr) \,\nu ( d x), \qquad t \in \R_+, 
\end{equation}
where $a,b \geq 0$, and $ \nu$ is a positive measure, called Lévy measure, satisfying  
$$ \int_0^{+\infty} \min ( 1 , x) \, \nu ( d x) < \infty . $$
\subsection{The original Gneiting class of space-time covariance functions}
Even though the  applications we have in mind concern space and time, all the theoretical background is actually valid in the $\R^k \times \R^l$ more general setting, with $k$ and $l$ two positive integers. Space-time corresponds to the particular case $l=1$. The Gneiting class of covariances on $\R^k \times \R^l$ \citep{Gneiting2002} involves two functions, usually denoted $\varphi$ and $\psi$ and associated with the ``spatial'' (on $\R^k$) and the ``temporal''  (on $\R^l$) structures, respectively: 
\begin{theorem}[(Gneiting, 2002)]
Let $\sigma > 0$, $\varphi(t), t \geq 0$, a completely monotone function on $\R_{+}$ and $\psi(t), t \geq 0$, a Bernstein function. Then,
\begin{equation}
    C(\bh,\bu) = \frac{\sigma^2}{\psi(\lvert\bu\rvert^2)^{k/2}}\, \varphi \left( \frac{\lvert\bh\rvert^2}{\psi(\lvert\bu\rvert^2)}\right), \qquad (\bh,\bu) \in \R^k \times \R^l,
   \label{eq:G}
\end{equation}
 is a covariance on $\R^k \times \R^l$.
 \label{theorem1}
\end{theorem}
Notice that the function defined in \eqref{eq:G} is fully symmetric, i.e. 
$C(\bh,\bu) = C(-\bh,\bu) = C(\bh,-\bu) = C(-\bh,-\bu)$.
Without loss of generality, we suppose from now on that $C(\bzero,0)=1$. Hence, following \cite{Gneiting2002}, one can assume that $\sigma^2=1$, $\psi(0)=a=1$ and that the measure $\mu$ defined in \eqref{eq:Laplace} is a probability measure. Unless specified otherwise, we shall also consider throughout that $\mu$ has no atom at 0, which implies that $\lim_{t \to \infty} \varphi(t) =0$.

\subsection{The extended Gneiting class}
\label{sec:extended}
The temporal structure is parameterized by the function $\psi$. From Gneiting's theorem, a sufficient condition on 
$\psi$ for the function given in \eqref{eq:G} to be a space-time covariance is that $\psi(t)$ must be a Bernstein function. Such a function leads to a variogram that is valid for any dimension $l \in \mathbb{N}^*$, a result that was established in \cite{zastavnyi2011characterization}. For the sake of completeness, we provide here a more direct proof.

\begin{proposition}
\label{lemma4}
For any positive integer $l$ and any Bernstein function $\psi$ on $\R_+$, the function $\gamma$ defined by 
$$\gamma(\bu) := \psi(\lvert\bu\rvert^2)-\psi(0) =\psi(\lvert\bu\rvert^2)- 1,\ \bu \in \R^l,$$ 
is a variogram on $\R^l$.
\end{proposition}
\begin{proof}
For fixed $s > 0$, the function $e^{-s (\psi(t)-\psi(0))}, t \ge 0$, is completely monotone on $\R_+$ \citep{Feller1966}. Equivalently, it is the Laplace transform of a probability measure $\nu_s$:
\begin{equation*}
e^{-s (\psi(t)-\psi(0))} =  \int_{\R^+}e^{-rt} \nu_s(dr).
\end{equation*}
Therefore, $e^{-s \gamma(\bu)}, \bu \in \R^l$, is a covariance on $\R^l$, since it is  a mixture of Gaussian covariances: 
\begin{equation*}
e^{-s \gamma(\bu)} = \int_{\R^+}e^{-r \lvert\bu\rvert^2} \nu_s(dr).
\end{equation*}
This implies that $\gamma(\bu), \bu \in \R^l$, is a variogram on $\R^l$ \citep{Schoenberg1938b, Chiles2012}.
\end{proof}

Accordingly, the Gneiting covariance \eqref{eq:G} belongs to the more general class of functions of the form
\begin{equation}
    C(\bh,\bu) = \frac{1}{(\gamma(\bu)+1)^{k/2}}\varphi\left( \frac{\lvert\bh\rvert^2}{\gamma(\bu)+1}\right), \qquad (\bh,\bu) \in \R^k \times \R^l,
    \label{eq:G2}
\end{equation}
where $\varphi$ is a completely monotone function on $\R_{+}$ and $\gamma$ a continuous variogram on $\R^l$. \cite{zastavnyi2011characterization} have shown that every member of this extended class is a valid covariance on $\R^k \times \R^l$. From a modeling point of view, the formulation \eqref{eq:G2} offers much more flexibility than \eqref{eq:G}, which only allows for isotropic variograms associated with Bernstein functions as per Proposition \ref{lemma4}. In particular, for space-time applications, any one-dimensional variogram is admissible, including nonmonotonic variograms. In the remainder of this work, we will adopt the  formulation \eqref{eq:G2}, hereafter referred to as the extended Gneiting class of covariance functions or the class of Gneiting-type covariance functions. The substitution approach presented in Section \ref{sec:substitution} will provide an alternative constructive proof of this result.

In this context, the Lévy representation \eqref{eq:btn} turns out to be unnecessarily restrictive. Similar to Bochner's theorem for covariance functions, the spectral representation of the variogram \citep{kmv61,ylm57} states that a continuous function $\gamma \neq 0$ on $\R^l$ is a variogram if and only if
\begin{equation}\label{eq:gmasr}
\gamma (\bu) = Q(\bu) + \int_{\R^l} \bigl[ 1 - \cos (\langle \bu, \bx \rangle) \bigr] \, \sm ( d\bx ), \qquad \bu \in \R^l,
\end{equation}
where $Q(\bu)$ is a nonnegative quadratic form  and where the spectral measure $\sm$ is positive, symmetric, without an atom at the origin, and satisfies
\begin{equation}\label{eq:bndsm}
\int_{\R^l} \frac{|\bx|^2 \, \sm ( d\bx )}{1 + |\bx|^2} < \infty.
\end{equation}
According to Proposition 4.5 in \cite{matheron1972leccon},  $Q(\bu)=0$ if and only if $\gamma(\bu)/|\bu|^2 \to 0$ as  $|\bu| \to \infty$. For the sake of simplicity and following common usage in spatial statistics, we will assume throughout that $Q(\bu)=0$.

The temporal covariance function associated with \eqref{eq:G2}
\begin{equation}
C_T(\bu) = C(\bzero,\bu) = \bigl( \gamma(\bu) + 1 \bigr)^{-k/2},
\label{eq:C_T}
\end{equation}
is not necessarily integrable. For example, this occurs when $\gamma(\bu)$ is bounded, since in this case $C_T(\bu)$ does not tend to zero as $\lvert \bu \rvert$ tends to infinity. This prevents us using properties of the Fourier transform to establish the existence of a spectral density in all generality. Theorem \ref{thm=dens-spec} is a characterization result, which provides {a} necessary and sufficient condition for the existence of a spectral density for $C$.

\begin{theorem}\label{thm=dens-spec} Assume that the spectral measure $\sm$ is absolutely continuous. Then the covariance $C$ has a spectral density if and only if  $ \sm(\R^l)=\int_{\R^l} \sm(d \bx)=+\infty$.
\end{theorem}

This theorem, proven in Appendix A, relies on the conditional decomposition of the spectral measure presented in Section \ref{sec:sampling_sm}.  A Gneiting covariance from the original class, i.e. defined by \eqref{eq:G}, always admits the representation \eqref{eq:gmasr} where $\sm$ is an absolutely continuous measure such that $\sm(\R^l)=\nu(\R_+)$, where $\nu$ is the Lévy measure defined in \eqref{eq:btn}. Hence Theorem 2 offers a full characterization of the subclass of original Gneiting covariances with a spectral density. It is interesting to make a link between Theorem 2 and bounded variograms. First of all, we establish the following proposition.
\begin{proposition}\label{prp:bnd}
A variogram $\gamma(\bu)$ is bounded and thus associated with a covariance function $c(\bu)$ with $\gamma(\bu)=c(\bzero)-c(\bu)$, if and only if $\sm(\R^l)=\int_{\R^l} \sm(d \bx) < +\infty$.

\label{prop:bounded}
\end{proposition}
\begin{proof}
First, assume that the symmetric measure $\sm$  is finite,  that is $\sm(\R^l)=A<+\infty$. Then, \eqref{eq:gmasr} becomes 
$$\gamma(\bu) = A - \int_{\R^l} \cos (\langle \bu, \bx \rangle) \sm (d \bx) =c(\bzero)-c(\bu)
$$
where by Bochner's theorem, $c(\bu)=\int_{\R^l} \cos (\langle \bu, \bx \rangle) \sm (d \bx)$ is a covariance function. The variogram $\gamma$ is thus bounded. 
Reciprocally, assume that the variogram  $\gamma$ is bounded, so that  $c(\bu) = B - \gamma(\bu)$ is a  covariance function for some finite value $B$; its Bochner representation is $c(\bu) = \int_{\R^l}  \cos (\langle \bu, \bx \rangle) \eta(d\bx)$ with  $\eta(\R^l) = B$. Hence, $\gamma(\bu) = B - c(\bu) = \int_{\R^l} \bigl( 1 -   \cos (\langle \bu, \bx \rangle)  \bigr) \eta(d\bx)$ and a direct comparison with \eqref{eq:gmasr} leads to $\eta(d\bx) = \sm (d\bx)$. Accordingly, 
$\sm (\R^l) = B < +\infty$. 
\end{proof}
Putting  Theorem \ref{thm=dens-spec} and Proposition \ref{prop:bounded} together, we establish the following: an extended Gneiting covariance function admits a spectral density  when $\gamma(\bu)$ is unbounded and its spectral measure $\sm$ is absolutely continuous. Examples of pairs $(\gamma,\sm)$ are given in Table \ref{tab:vario}. 

\medskip

We now recall Bochner's theorem \citep{Bochner1955}, according to which a continuous function $C$ on $\R^k \times \R^l$ is positive semi-definite, hence a covariance function, if and only if
\begin{equation}
    C(\bh,\bu) =  \int_{\R^k} \int_{\R^l} e^{i\langle \bomega,\bh \rangle+ i \langle \btau, \bu \rangle }\, F(d\bomega,d\btau), \qquad (\bh,\bu) \in \R^k \times \R^l,
    \label{eq:BochnerF}
\end{equation}
where $F$ is a nonnegative finite, symmetric, measure on $\R^k \times \R^l$, known as the spectral measure of $C$. Note that since the extended Gneiting class is fully symmetric, the spectral measure $F$ must also be fully symmetric. We finish this Section with a lemma that will be useful for the simulation algorithm proposed in Section \ref{sec:substitution}.
\begin{lemma}
The extended Gneiting covariance function defined in \eqref{eq:G2} can be written as follows:
\begin{equation}
C(\bh,\bu) = 
\frac{1}{{(2\pi)}^{k/2}} \int_{\R_+} \int_{\R^k} \cos \left(\sqrt{2r} \, \langle \tilde{\bomega},\bh \rangle \right) \exp \left( -\frac{\lvert\tilde{\bomega}\rvert^2 (\gamma(\bu)+1)}{2} \right) d\tilde{\bomega} \, \mu(dr), \qquad (\bh,\bu) \in \R^k \times \R^l,
\label{eq:emery111}
\end{equation}
where $\mu$ is the measure associated with $\varphi$, as defined in \eqref{eq:Laplace}.
\end{lemma}

\begin{proof}
The inner integral in the right-hand side of \eqref{eq:emery111} is, up to a multiplicative factor, the Fourier transform at $\sqrt{2r} \, \bh$ of the function $\tilde{\bomega} \mapsto \exp \left( -\frac{\lvert\tilde{\bomega}\rvert^2 (\gamma(\bu)+1)}{2} \right)$ and is equal to $\left(\frac{2\pi}{\gamma(\bu)+1}\right)^{k/2} \exp \left( -\frac{r \lvert\bh\rvert^2}{\gamma(\bu)+1} \right)$. Formula \eqref{eq:emery111} then follows from \eqref{eq:Laplace} and \eqref{eq:G2}.
\end{proof}
For the sake of clarity, and because most applications are in a space-time context, the presentation will be made in the particular case when $l=1$. Extensions to the more general case when $l \in \N^*$ will be discussed in Section \ref{sec:approaches}.

%
\section{Spectral Approach}
\label{sec:spectral}
%

Since we assumed $C(\bzero,0)=1$, the spectral  measure $F$ defined in \eqref{eq:BochnerF} is a probability measure. 
Let $ (\bOmega, \Tau)$ be a spectral vector distributed  {as} $F$. Let also $\Phi$ be a random phase that is uniform 
on $(0,2 \pi)$ and independent of $(\bOmega, \Tau)$. In its most basic form, the spectral method rests on the fact 
that the random field defined by 
\begin{equation}\label{eq:spectral0}
 Z(\bx,t) = \sqrt{2} \cos ( \langle \bOmega , \bx \rangle + \Tau t + \Phi ), \qquad ( \bx,t) \in \R^k \times \R,
 \end{equation}
has a zero mean and covariance $C$ \citep{Shinozuka1971}. By introducing a multiplicative factor $\sqrt{-2 \ln (U)}$, 
where $U$ is an independent uniform variable on $(0,1)$, the Box-Muller transformation \citep{Box1958} ensures that all marginal 
distributions of $Z$ are standard Gaussian. To go further and obtain a random field whose finite-dimensional distributions (not only the marginals) are approximately Gaussian, one option is to add and rescale many independent copies of the form \eqref{eq:spectral0}: 
\begin{equation}\label{eq:spectral00}
\tilde{Z} (\bx,t) = \sum_{j=1}^p \sqrt{\frac{-2 \ln(U_j)}{p}} \cos \left( \langle \bOmega_j, \bx \rangle + 
\Tau_j t + \Phi_j \right), \qquad (\bx,t) \in \R^k \times \R,
\end{equation}
where $p$ is a positive integer and $\bigl( (\bOmega_j,\Tau_j,\Phi_j,U_j): j = 1, \dots \, , p \bigr)$ are independent 
copies of $(\bOmega,\Tau,\Phi,U)$. Because of the central limit theorem, the finite-dimensional distributions of $\tilde 
Z$ tend to become multivariate Gaussian as $p$ tends to infinity \citep{Lantu2002}. Algorithm \ref{agm:stlsmn}, presented 
below, corresponds to this approach. In the rest of this work, ${\cal U} (a,b)$ denotes a uniform random variable on 
the interval $(a,b)$ and "$\sim$" means "is distributed as".

\medskip

\begin{algorithm}[htb]
\caption{Spectral simulation of a Gaussian random field with Gneiting type space-time covariance} 
\label{agm:stlsmn}
\begin{algorithmic}[1]
\REQUIRE $k$ and $p$
\REQUIRE $ F  $
\FOR{$j=1$ to $p$} 
\STATE Simulate $ (\bOmega_j, \Tau_j) \sim F$;
\STATE Simulate  $\Phi_j \sim {\cal U}(0, 2\pi)$;
\STATE Simulate  $U_j \sim {\cal U}(0,1)$; 
\ENDFOR
\STATE For each $(\bx,t) \in \R^k \times \R$ return $ \tilde{Z} (\bx,t) = \sum_{j=1}^p \sqrt{-2 \ln(U_j) / p} 
\, \cos \bigl( \langle \bOmega_j, \bx \rangle + \Tau_j t + \Phi_j \bigr) $.
\end{algorithmic}
\end{algorithm}
To be successfully implemented, Algorithm 1 requires the spectral measure to be simulated. This issue is discussed in the next section. 

\subsection{Sampling the spectral measure} \label{sec:sampling_sm}
The function $\varphi$ introduced in \eqref{eq:Laplace} is a mixture of exponential functions. The mixture 
parameter $r$ can be seen as deriving from a latent random variable $R$ with distribution $\mu$. Accordingly, 
a covariance $C$ belonging to the extended Gneiting class \eqref{eq:G2} is also a mixture of basic covariance functions  
$$ C (\bh, u) = \int_0^{+\infty} \frac{1}{\bigl(1+\gamma(u) \bigr)^{k/2}} \, \exp \left( - \frac{r \lvert 
\bh \rvert^2} {1+\gamma(u)} \right) \, \mu (d r) \equiv \int_0^{+\infty} C (\bh, u \mid r ) \, \mu (d r) , $$
and their spectral measures are related by the formula
$$  F ( d \bomega , d \tau ) = \int_0^{+\infty} F ( d \bomega , d \tau \mid r) \, \mu (d r) . $$
The proposed simulation algorithm relies on the factorization of the spectral measure $F ( d \bomega,
d \tau \mid r )$ into a spatial component and a conditional temporal component:
$$ F ( d \bomega, d \tau \mid r ) = F^{\phantom{|}}_S ( d \bomega \mid r) \, F^{\phantom{|}}_T ( d \tau \mid 
\bomega , r). $$
The following theorem, the proof of which is deferred to Appendix A, makes the expression of these components available 
via their Fourier transforms. 

\begin{theorem}\label{thm:cond}
Consider a Gneiting-type covariance function as given in \eqref{eq:G2}. The two following assertions holds.
\begin{enumerate}
\item If $ \bOmega (r) \sim  F^{\phantom{|}}_S ( d \bomega \mid r) $, then
\begin{equation}\label{eq:ftspl}
\E \left[ e^{\textstyle i \, \langle \bh, \bOmega (r) \rangle } \right] = e^{ \textstyle - r \, \vert \bh 
\vert^2} .
\end{equation}
\item If $ \Tau (\bomega, r) \sim F^{\phantom{|}}_T ( d \tau \mid \bomega , r) $, then 
\begin{equation}\label{eq:fttpl}
\E \left[ e^{ \textstyle i \, u \, \Tau(\bomega,r) } \right] = \exp \left(-\frac{\vert \bomega \vert^2 \, \gamma (u) 
}{4 r} \right) \qquad \bomega\text{-a.e.}
\end{equation}
\end{enumerate}
\end{theorem}

Eq. \eqref{eq:ftspl} shows that all the components of $\bOmega (r)$ are independent and normally distributed with zero mean and variance $2r$. The spatial spectral measure can therefore be easily simulated. The simulation of the conditional temporal spectral measure is more tricky. An adaptation of the shot-noise approach developed by \cite{bdn82} is proposed here. Consider the spectral measure $\sm$ of the variogram  introduced in Section \ref{sec:extended} (eqs. \eqref{eq:gmasr} and \eqref{eq:bndsm}).
As shown in Appendix A, 
it is possible to define a positive and  locally integrable function $\theta$ defined on $\R_+$, as well as a family of probability measures $\bigl( \sm_t: t > 0 \bigr)$ on $\R_+$, such that
\begin{equation}\label{eq:dcn}
\sm (d x) = \int_0^{+\infty} \sm_t ( d x) \, \theta (t) \, dt.
\end{equation}
Consider now a Poisson point process $\bigl( \Tau_n: n \geq 1 \bigr)$ with intensity function 
$ \lambda (t) = \frac{\lvert \bomega \rvert^2 \, \theta (t)}{4 r}$ on $\R_+$. Let us  independently assign to each $\Tau_n$, 
a random variable $X^{\phantom{|}}_{\Tau_n}$ distributed according to $\sm^{\phantom{|}}_{\Tau_n}$. Then, it is shown 
in Appendix A that the distribution of the random variable $\sum_{n \geq 1} X^{\phantom{|}}_{\Tau_n}$ coincides 
with that of $\Tau(\bomega,r)$.

This algorithm is generic, in the sense that it is applicable to any variogram $\gamma$. Now, it should be pointed out that it can be only approximately implemented if the Poisson point process contains infinitely many points, which occurs when $\sm(\R) = + \infty$, or {equivalently}  from Prop. \ref{prp:bnd}, when $\gamma$ is unbounded. \cite{bdn82} made a number of recommendations about the effective number of Poisson points to simulate and the way to approximate 
the remainder. In contrast to this, if the spectral measure of the variogram is integrable, then the number of points of the Poisson process is almost surely finite, which makes the algorithmic implementation possibly exact. The counterpart 
is that this Poisson number may be equal to zero with a nonzero probability. In such a case, the distribution of $\Tau(\bomega,r)$ has an atom at $0$, which is equivalent to saying that $C(0,u)$ does not tend to zero as $u$ tends to infinity.

To apply this algorithm, the spectral measure of the variogram is explicitly needed. This is a limitation because such spectral measures are not always available. Table \ref{tab:vario} provides a list of variograms and their associated spectral measures. Note however that "universal" algorithms have been developed to simulate monovariate distributions starting from their Fourier transforms \citep{dve81,bripti15}, under some conditions on the characteristic functions. These algorithms could be used to simulate directly $\Tau(\bomega,r)$ without knowing explicitly the spectral measure $\sm$. We do not pursue this route in this work since, as we will see in the next section, for many variograms, specific simulation algorithms  can be conceived. Typical examples include the variogram $\gamma (u) = \bigl( a \lvert u \rvert^\alpha + 1 
\bigr)^\beta$ \citep{Gneiting2002}.

\medskip

\begin{table}
\caption{Spectral measures corresponding to selected one-dimensional variograms. \label{tab:vario}}
\begin{center}
\begin{tabular}{ccc}
\hline\hline
$ \gamma (u) $ & $ \dpst \frac{\sm (dx)}{dx} $ & \text{Restriction} \\[0.2cm] 
\hline
$ \dpst \absu^\alpha $ & $ \dpst \frac{- 2 \, \Gamma (\alpha)}{ \Gamma ( \alpha / 2) \, \Gamma ( - \alpha /2)} \, 
\frac{1}{\absx^{1 + \alpha}} $ & $ 0 < \alpha < 2 $ \\[0.3cm]
 $ \dpst \absu -1 + e^{- \absu} $ & $ \dpst \frac{1}{\pi} \, \frac{1}{x^2 \, ( 1 + x^2)} $ &  \\[0.3cm]  
$ \dpst \begin{cases}
    \absu        & \text{if $\absu < 1$} \\
    2 \absu - 1  & \text{if $\absu \geq 1$}
   \end{cases} $
  & $ \dpst \frac{1}{\pi} \, \frac{1 + \cos x}{x^2} $ & \\[0.5cm]
 $ \dpst \begin{cases}
    u^2 ( 3 - \absu )   & \text{if $\absu < 1$} \\
    3 \absu - 1)        & \text{if $\absu \geq 1$}
   \end{cases} $
   & $ \dpst \frac{6}{\pi} \frac{ 1 - \cos x}{x^4} $  & \\[0.3cm]
 $ \dpst \ln ( 1 + u^2) $ & $ \dpst \frac{ \exp ( - \absx ) }{\absx} $ & $ $ \\[0.3cm]
$ \dpst 8 \sqrt{\pi} \, \Bigl( \sinh \frac{\argsinh u}{4} \Bigr)^2 $ & $ \dpst 
\frac{ \exp ( - \absx ) }{\absx^{3/2} } $ & $ $ \\[0.3cm]
 $ \dpst 2 \absu \arctan \absu - \ln ( 1 + u^2 ) $ &  $ \dpst \frac{ \exp ( - \absx ) }{\absx^2 } $ & $ $ \\[0.3cm]

 $ \dpst \frac{8 \sqrt{\pi}}{3} \left( 1 - (1 + u^2)^{3/4} \cos \frac{3 \, \arctan u}{2} \right) $  
 & $ \dpst \frac{ \exp ( - \absx ) }{\absx^{5/2} } $ & $ $ \\[0.3cm]
\hline\hline
\end{tabular}
\end{center}
\end{table}

\subsection{Illustrations}
Three examples are presented in $\R^2 \times \R$, all based on the same completely monotone function $\varphi (t) = 
\exp ( - r t)$ with $ r= 0.01$, but with different variograms $\gamma (u)$. All  simulations have been obtained using $p=5000$ basic cosine waves. They are displayed on a $300 \times 200$ grid with a unit square mesh size, using the same color scale 
ranging from $-4$ to $+4$. Six consecutive images, separated by short intervals of $0.2$ time units, allow to keep  track of the evolution of the large-value zones with time.

\subsubsection{First example}
The linear variogram $\gamma (u) = b \lvert u \rvert$ is certainly one of the simplest variograms that can be 
considered. Its spectral measure is proportional to the Lebesgue measure. The associated space-time covariance 
function is 
$$ C(\bh,u) = \frac{1}{1 + b \lvert u \vert} \, \exp \left( - \frac{- r \, \lvert \bh \rvert^2}{1 + b \lvert u \rvert} 
\right) , $$
which includes a spatial Gaussian covariance  and a temporal hyperbolic covariance. The conditional temporal frequencies follow a Cauchy distribution. The simulation of Fig. \ref{fig:lin} has been obtained by taking $b=1$. 

\smallskip

\begin{figure}[htp]
    \begin{center}
    \includegraphics[width=4.5cm]{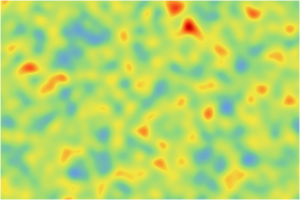} \hspace{0.1cm} \includegraphics[width=4.5cm]{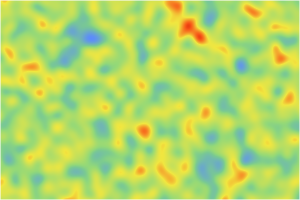} \hspace{0.1cm} \includegraphics[width=4.5cm]{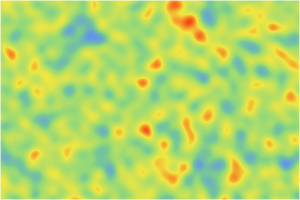}
    \end{center}
    \vspace{-0.4cm}
    \begin{center}
    \includegraphics[width=4.5cm]{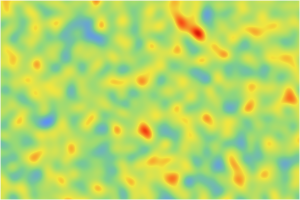} \hspace{0.1cm} \includegraphics[width=4.5cm]{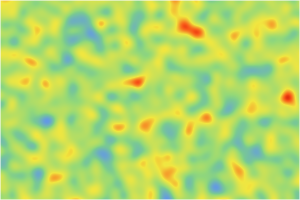} \hspace{0.1cm} \includegraphics[width=4.5cm]{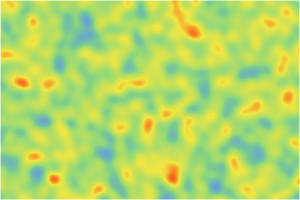}
    \end{center}
    \caption{Simulation of a Gneiting model with linear variogram.} 
\label{fig:lin}
\end{figure}

\medskip

\subsubsection{Second example}

The second example considers the logarithmic variogram $ \gamma (u) = \ln ( a^2 + u^2) / \ln a^2 - 1$. The 
corresponding space-time covariance function is equal to
$$ C ( \bh , u ) = \frac{\ln a^2}{\ln ( a^2 + u^2)} \, \exp \Bigl( - \frac{r \, \lvert \bh \rvert^2 \, \ln a^2}
{\ln (a^2 + u^2)} \Bigr). $$
Compared to the first example, the temporal covariance function
$$ C_T (u ) = \frac{\ln a^2}{\ln ( a + u^2)} $$
vanishes at infinity at very slow rate. The spectral measure of $\gamma$ can be derived from the 5th entry of Table \ref{tab:vario},
$$ \sm (d x) = \frac{\exp (- a \lvert x \rvert)}{ \lvert x \rvert \, \ln a^2}\,{dx},$$
which is not integrable in agreement with Proposition \ref{prop:bounded}. To design a simulation algorithm, the generic 
approach can be applied (see the details in Appendix A), but in the present case a more direct approach is also possible. 

Let us start with the Fourier transform \eqref{eq:fttpl} of $ \Tau (\bomega, r)$. Replacing $\gamma$ by its 
expression and putting $\lambda = \frac{\lvert \bomega \rvert^2}{4 r \ln a^2} $, one obtains
$$ \E \left[ e^{ \textstyle i \, u \, \Tau(\bomega,r) } \right] = \left( \frac{a^2}{a^2 + u^2} \right)^\lambda. $$ 
The right-hand side member of this equation can be seen as the Laplace transform at $u^2$ of a gamma distribution 
with parameter $\lambda$ and index $a^2$. Accordingly, one can write
$$ \E \left[ e^{ \textstyle i \, u \, \Tau(\bomega,r) } \right] = \frac{a^{ 2 \lambda}}{\Gamma (\lambda)} 
\int_0^{+\infty} e^{ \textstyle - x u^2} \, e^{ \textstyle - a^2 x} x^{\lambda - 1} \, dx . $$
Since $ \exp ( - x u^2 )$ is the Fourier transform of a centered Gaussian variable with variance $2 x$, it is  obtained that the distribution of $ \Tau (\bomega, r)$ is a gamma mixture mixture of Gaussian distributions. The explicit 
description is detailed in the following algorithm. Below, and in the rest of this work, ${\cal N}(\mu,\sigma^2)$ 
denotes the Gaussian (normal) distribution with expectation $\mu$ and variance $\sigma^2$, and ${\cal G} (\lambda,a^2)$ 
denotes the gamma distribution with shape parameter $\lambda$ and scale parameter $a^2$.

\smallskip

\begin{algorithm}
\caption{Sampling from the conditional spectral distribution with logarithmic variogram} 
\label{agm:log}
\begin{algorithmic}[1]
\REQUIRE $a > 0 , \lambda = \frac{\lvert \bomega \rvert^2}{4 r \ln a^2} $
\STATE Simulate $ X \sim {\cal{G}} (\lambda, a^2)$ and $ Y \sim {\cal N} (0,1)$;
\STATE Return $\Tau (\bomega, r) = Y \, \sqrt{ 2 X }$. 
\end{algorithmic}
\end{algorithm}

\smallskip

The simulation shown in Fig. \ref{fig:log} has been obtained using this algorithm with $a=2.06$, a value chosen so that 
the temporal covariance functions of the first two examples take the same value at the time lag $0.2$ between 
successive images. 

\begin{figure}[htp]
    \begin{center}
    \includegraphics[width=4.5cm]{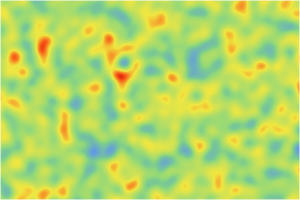} \hspace{0.1cm} \includegraphics[width=4.5cm]{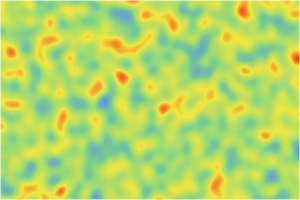} \hspace{0.1cm} \includegraphics[width=4.5cm]{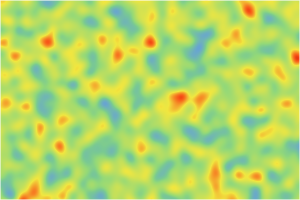}
    \end{center}
    \vspace{-0.4cm}
    \begin{center}
    \includegraphics[width=4.5cm]{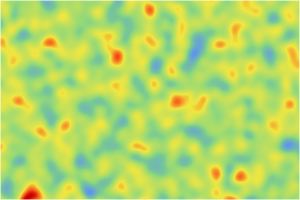} \hspace{0.1cm} \includegraphics[width=4.5cm]{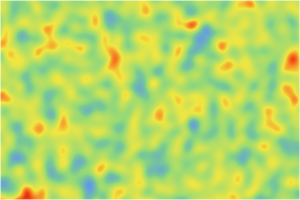} \hspace{0.1cm} \includegraphics[width=4.5cm]{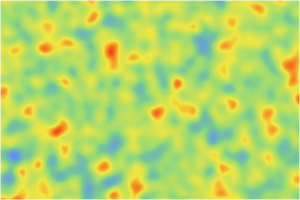}
    \end{center}
    \caption{Simulation of a Gneiting model associated with a logarithmic variogram.} 
\label{fig:log}
\end{figure}

\subsubsection{Third example}

The third example is the variogram  used in \cite{Gneiting2002} for modeling the Irish wind data set, namely$\gamma(u) = (a|u|^\alpha +1)^\beta -1$ with $a > 0$, $ 0 < \alpha \leq 2$ and $0 < \beta \leq 1$. It leads to the space-time covariance function
$$ C ( \bh , u ) = \frac{1}{(a |u|^\alpha + 1)^\beta} \, \exp \Bigl( - \frac{r \, \lvert \bh \rvert^2 }{(a 
|u|^\alpha + 1)^\beta} \Bigr). $$
In particular, the associated temporal covariance $C_T (u ) = (a |u|^\alpha + 1)^{-\beta}$ belongs to the Cauchy 
class \citep{gneiting2004stochastic} and can exhibit a wide range of behaviors. The  parameter $\alpha$ governs 
the behavior at the origin (trajectories become smoother as 
$\alpha$ increases), while the product $\alpha \beta$ controls the rate of decay at long times (trajectories have a longer memory as $\alpha \beta$ increases). 
A simulation of a Gneiting model with $a=1$, $\alpha =1$ and $\beta = 0.5$ is depicted on Fig. \ref{fig:gtg}. 

\medskip

\noindent The spectral measure of this variogram is unknown, but, here too, the conditional spectral measure can 
be simulated directly. Let us start again with the Fourier transform \eqref{eq:fttpl} of $ \Tau (\bomega, r)$. 
Replacing $\gamma$ by its expression, and putting $\lambda = \frac{\lvert \bomega \rvert^2}{4r}$, the Fourier transform becomes 
$$ \E \left[ e^{ \textstyle i \, u \, \Tau(\bomega,r) } \right] = \exp \left( - \lambda \bigl( ( a |u|^\alpha + 
1)^\beta - 1 \bigr) \right). $$ 
Up to the factor $\exp \, (\lambda)$, the right-hand side member is the Laplace transform at $ \lambda^{1 / \beta} 
\, \bigl( a |u|^\alpha + 1 \bigr) $ of a unilateral stable distribution $ {\cal S}^+ (\beta)$ with stability 
index $\beta$ \citep{zlv86}. Denoting by $f_\beta$ its probability density function, one can write
$$ \E \left[ e^{ \textstyle i \, u \, \Tau(\bomega,r) } \right] = \exp (\lambda) \int_0^{+\infty} \exp \left( - 
\lambda^{1 / \beta} \, \bigl( a |u|^\alpha + 1 \bigr) \, s \right) \, f_\beta (s) \, d s. $$ 
Now, another grouping of factors gives 
$$ \E \left[ e^{ \textstyle i \, u \, \Tau(\bomega,r) } \right] = \int_0^{+\infty} g_\beta (s) \, \exp \left( - 
\lambda^{1 / \beta} a |u|^\alpha s \right) \, d s, $$ 
where

\begin{itemize}
    \item  $g_\beta (s) = \exp \bigl(\lambda - \lambda^{1 / \beta} \, s \bigr) \, f_\beta (s) $ is the density of 
    another distribution, denoted by $ {\cal S}^+ (\beta,\lambda^{1 / \beta})$ and called exponentially tilted 
    unilateral stable distribution with stability index $\beta$ and tilting parameter $\lambda^{1 / \beta}$ 
    \citep{sto99}; 
    \item $ \exp \left( - \lambda^{1 / \beta} a |u|^\alpha s \right) $ is the Fourier transform at $ \bigl( 
    \lambda^{1/\beta} a s )^{1 / \alpha}$ of a bilateral stable distribution ${\cal S} (\alpha)$ with stability 
    index $\alpha$ \citep{lvy25,kcelvy36}. 
\end{itemize}{}

It thus appears that the distribution of $ \Tau (\bomega, r)$ is a mixture of bilateral stable distributions.  
Algorithm \ref{agm:gtg} makes this assertion more precise.

\medskip

\begin{algorithm}[h]
\caption{Sampling from the conditional spectral distribution with Cauchy temporal covariance}
\label{agm:gtg}
\begin{algorithmic}[1]
\REQUIRE $a > 0 , \, 0 < \alpha \leq 2, \, 0 < \beta \leq 1, \, \lambda = \frac{\lvert \bomega \rvert^2}{4r}$
\STATE Simulate $ S \sim {\cal S}^+ (\beta , \lambda^{1 / \beta}) $;
\STATE Simulate $ T \sim {\cal S} (\alpha)$; 
\STATE Return $\Tau ( \bomega,r) = T \, ( S a \lambda^{1/\beta} )^{1 / \alpha}$. 
\end{algorithmic}
\end{algorithm}

There remains to see how to simulate those stable distributions. $ {\cal S} (\alpha)$ can be simulated using the 
fast algorithm by \cite{cbsetal76} that generalizes the Box-Muller algorithm to simulate normal distributions and 
has become a standard. \cite{brx99} proposes an algorithm to simulate $ {\cal S}^+ (\beta , \lambda^{1 / \beta})$ 
based on a rejection from ${\cal S}^+ (\beta)$. However, this algorithm may suffer from a high rejection rate for 
large values of $ \lambda^{1 / \beta}$. This prompted \cite{dve09} to propose a double rejection technique that 
possesses a uniform and limited rejection rate. 

\begin{figure}[htp]
    \begin{center}
    \includegraphics[width=4.5cm]{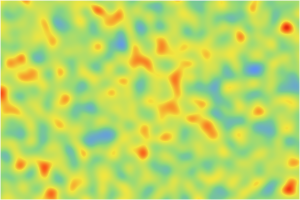} \hspace{0.1cm} \includegraphics[width=4.5cm]{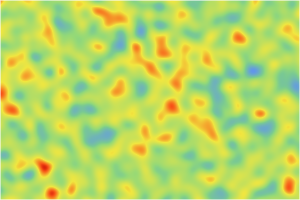} \hspace{0.1cm} \includegraphics[width=4.5cm]{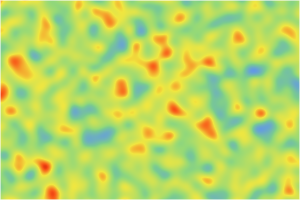}
    \end{center}
    \vspace{-0.4cm}
    \begin{center}
    \includegraphics[width=4.5cm]{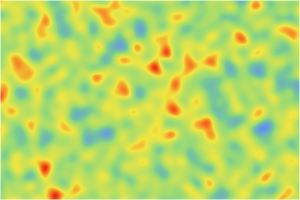} \hspace{0.1cm} \includegraphics[width=4.5cm]{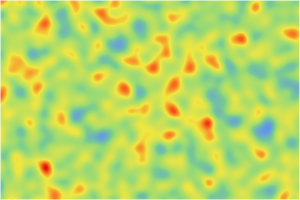} \hspace{0.1cm} \includegraphics[width=4.5cm]{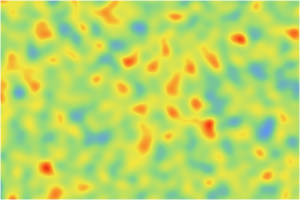}
    \end{center}
    \caption{Simulation of a Gneiting model associated with a Cauchy temporal covariance $C_T(u)$.} 
\label{fig:gtg}
\end{figure}
%

%
\section{Substitution approach}
\label{sec:substitution}
%

%
\subsection{Proposal}

Instead of simulating a temporal frequency according to the conditional measure $F^{\phantom{|}}_T ( d \tau \mid \bomega , r)$, here we simulate an intrinsic Gaussian random field $W(t)$ on $\R$ with variogram $\gamma$, without making reference to its spectral representation \eqref{eq:gmasr}. Specifically, consider the random field $Z$ defined in $\R^k \times \R$ as follows:

\begin{equation}
Z(\bx,t) = \sqrt{-2 \ln(U)} \cos \left(\sqrt{2R} \, \langle \tilde{\bOmega}, \bx \rangle + \frac{\lvert\tilde{\bOmega}\rvert}{\sqrt{2}} W(t) + \Phi \right), \qquad (\bx,\bt) \in \R^k \times \R,
\label{eq:substitution1}
\end{equation}
where:
\begin{itemize}
    \item $R$ is a nonnegative random variable with probability measure $\mu$;

    \item $\tilde{\bOmega}$ is a Gaussian random vector of $k$ independent components with zero mean and unit variance;
    
    \item $U$ is a random variable uniformly distributed in $(0,1)$, independent of $(R,\tilde{\bOmega})$; 

    \item $\Phi$ is a random variable, independent of $(R,\tilde{\bOmega},U)$, uniformly distributed in $(0, 2\pi)$;
    
    \item $W(t),\ t \in \R$, is an intrinsic random field on $\R$ with variogram $\gamma$ and Gaussian increments, independent of $(R,\tilde{\bOmega},U,\Phi)$.
\end{itemize}

\begin{theorem}
The random field $Z$ defined in \eqref{eq:substitution1} is second-order stationary in $\R^k \times \R$, with zero mean and Gneiting-type covariance function as given in \eqref{eq:G2}.
\label{theorem2}
\end{theorem}

The proof of this theorem is given in Appendix B.  The random field $Z$ defined in  \eqref{eq:substitution1} is a particular case of a substitution random field, obtained by combining a directing function $D$ with stationary increments in $\R^k \times \R$ and a stationary coding process $X$ in $\R$ \citep{Lantu1991, Lantu2002}, a construction that generalizes the subordination approach introduced by \cite{Feller1966}. Here, the directing function is the sum of a space-dependent linear drift, which has stationary increments in $\R^k$, and a time-dependent random field with stationary increments in $\R$:
\begin{equation*}
    D(\bx,t) = \sqrt{2R} \, \langle \tilde{\bOmega}, \bx \rangle + \frac{\lvert\tilde{\bOmega}\rvert}{\sqrt{2}} W(t), \qquad (\bx,t) \in \R^k \times \R. 
\end{equation*}
As for the coding process, it is a cosine function in $\R$ with constant frequency $(2\pi)^{-1}$, random phase $\Phi$ uniformly distributed in $(0, 2 \pi)$ and random amplitude $\sqrt{-2 \ln(U)}$:
\begin{equation*}
    X(d) = \sqrt{-2 \ln(U)} \cos \left( d + \Phi \right), \qquad d \in \R. 
\end{equation*}
Such a coding process is stationary and has a Gaussian marginal distribution. Because $D$ and $X$ are independent, the substitution random field $Z = X \circ D$ defined in \eqref{eq:substitution1} inherits several properties of the coding process \citep{Lantu2002}, in particular it is stationary and has the same Gaussian marginal distribution as $X$. The latter property (Gaussian marginal) can also be proven on the basis of the Box-Muller transformation \citep{Box1958}.

There is actually more, as the coding process is a Gaussian random field in $\R$. This can be proven by observing that any weighted sum of variables $X(d_1), \cdots, X(d_j)$ has a Gaussian distribution, since it can be written under the form $a \sqrt{-2 \ln(U)} \cos \left( b + \Phi \right)$, with deterministic terms $a$ and $b$ that depend on the chosen weights and time instants $t_1, \cdots, t_j$. In particular, the bivariate distributions of $X$ are bi-Gaussian and have an isofactorial representation with Hermite polynomials as the factors \citep{Lancaster1957, Chiles2012}. The substitution construction therefore ensures that $Z$ also has bivariate distributions with Hermite polynomials as the factors \citep{Matheron1989, Lantu2002}, which are nothing else than mixtures of bi-Gaussian distributions \citep{Matheron1976, Chiles2012}. Note the similarity of this construction with the substitution models proposed by \cite{Matheron1982} and \cite{Emery2008}.

\subsection{Simulation algorithm}

The random field $Z$ defined in \eqref{eq:substitution1} is centered and second-order stationary. Its covariance function belongs to the Gneiting class \eqref{eq:G2} and its marginal distribution is Gaussian at any point $(\bx,t) \in \R^k \times \R$. 
To obtain a random field whose finite-dimensional distributions are approximately Gaussian,
one can simulate a large number of independent random fields $Z_j, j = 1, \dots, p$, as in \eqref{eq:substitution1}, and set:
\begin{equation}
\tilde{Z}(\bx,t) = \sum_{j=1}^p \sqrt{\frac{-2 \ln(U_j)}{p}} \cos \left(\sqrt{2 R_j} \langle \tilde{\bOmega}_j, \bx \rangle + \frac{\lvert\tilde{\bOmega}_j\rvert}{\sqrt{2}} W_j(t) + \Phi_j \right), \qquad (\bx,t) \in \R^k \times \R,
\label{eq:substitution3}
\end{equation}
where $\{ (R_j,\tilde{\bOmega}_j,U_j,\Phi_j,W_j): j = 1, \dots, p \}$ are independent copies of $(R,\tilde{\bOmega},U,\Phi,W)$. 

\medskip

Algorithm \ref{algo:substitution} described below can be used for simulating $\tilde{Z}$. The simulation of the intrinsic random field $ W_j$ at line 6 can be done by the covariance matrix decomposition method applied to the increment $W_j(t)-W_j(0)$, by fixing $W_j(0) = 0$ and using the nonstationary covariance $\gamma(t) + \gamma(t') - \gamma(t-t')$ \citep{Davis1987}. This is possible as long as the number of time instants considered for the simulation is not too large (less than a few tens of thousands). For larger numbers, other simulation methods are applicable, such as circulant embedding matrices with FFT \citep{wood1994simulation} or the Gibbs propagation algorithm \citep{Arroyo2015}. 

\begin{algorithm}
\caption{Substitution algorithm}
\label{algo:substitution}
\begin{algorithmic}[1]
\REQUIRE $\mu$ and $\gamma(u)$ 
\REQUIRE $p$
\FOR{$j=1$ to $p$}
\STATE Simulate $R_j \sim \mu$;
\STATE Simulate $\tilde{\bOmega}_j \sim {\cal N}_k(\bzero,\bI_k)$;
\STATE Simulate  $\Phi_j \sim {\cal U}(0, 2\pi)$;
\STATE Simulate  $U_j \sim {\cal U}(0, 1)$;
\STATE Simulate an independent intrinsic random field $W_j$ with Gaussian increments and  variogram $\gamma(u)$.
\ENDFOR
\STATE Compute the simulated random field at any target location $(\bx,t) \in \R^k \times \R$ as per \eqref{eq:substitution3}.
\end{algorithmic}
\end{algorithm}

As an illustration, consider the following covariance functions over $\R^2 \times \R$:
\medskip
$$ C_1 ( \bh , u ) = \frac{1}{\sqrt{1 + |u|}} \, \exp \Bigl( - \frac{0.01 \, \lvert \bh \rvert^2 }{\sqrt{1 + |u|}} \Bigr), $$
$$ C_2 ( \bh , u ) = \frac{1}{\sqrt{1 + |u|}} \, \exp \Bigl( - \frac{0.01 \, \lvert \bh \rvert }{\sqrt{1 + |u|}} \Bigr), $$
which belong to the Gneiting class \eqref{eq:G2} with  $\varphi_1 (t) = e^{- 0.01 t}$, $\varphi_2 (t) = e^{- 0.01 \sqrt{t}}$ and $\gamma(u) = (1 + |u|)^{0.5}-1$. Realizations of Gaussian random fields possessing these covariances are displayed in Figs. \ref{fig:pow} and \ref{fig:pow2} on a $300 \times 200$ spatial domain with unit square mesh size, for six consecutive time instants separated $0.2$ time unit. The simulation has been 
obtained with Algorithm \ref{algo:substitution},
by using $p=5000$ and the covariance matrix decomposition approach to simulate the random fields $W_j$ at the six time instants of interest. Concerning the second example, the probability measure associated with $\varphi_2$ is that of a random variable $R$ obtained by square rooting a gamma random variable with shape parameter $0.5$ \citep{emery2008tb}. 

\smallskip

\begin{figure}[htp]
    \begin{center}
    \includegraphics[width=4.5cm]{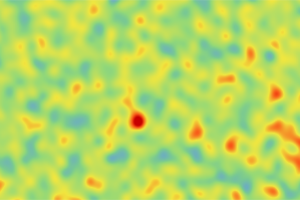} \hspace{0.1cm} \includegraphics[width=4.5cm]{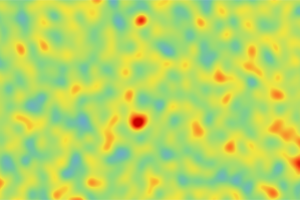} \hspace{0.1cm} \includegraphics[width=4.5cm]{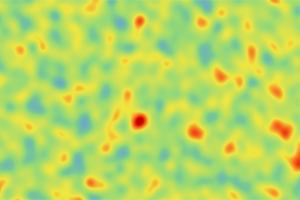}
    \end{center}
    \vspace{-0.4cm}
    \begin{center}
    \includegraphics[width=4.5cm]{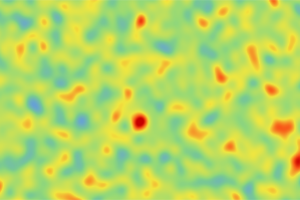} \hspace{0.1cm} \includegraphics[width=4.5cm]{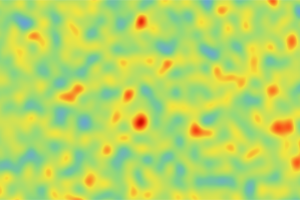} \hspace{0.1cm} \includegraphics[width=4.5cm]{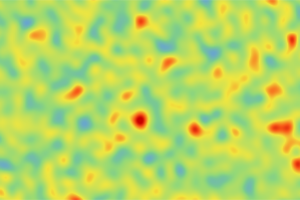}
    \end{center}
    \caption{Simulation of a Gneiting model associated with a Gaussian spatial covariance and a power variogram $\gamma(u)$.} 
    \label{fig:pow}
\end{figure}

\medskip

\begin{figure}[htp]
    \begin{center}
    \includegraphics[width=4.5cm]{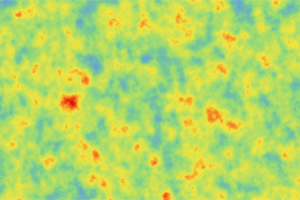} \hspace{0.1cm} \includegraphics[width=4.5cm]{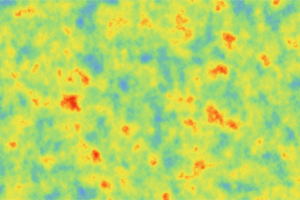} \hspace{0.1cm} \includegraphics[width=4.5cm]{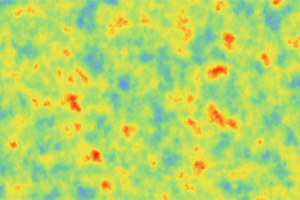}
    \end{center}
    \vspace{-0.4cm}
    \begin{center}
    \includegraphics[width=4.5cm]{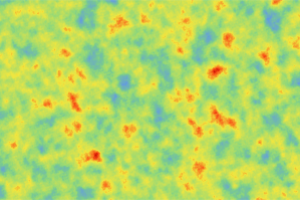} \hspace{0.1cm} \includegraphics[width=4.5cm]{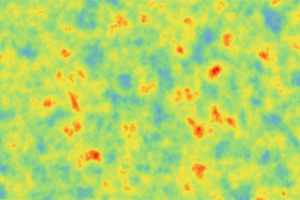} \hspace{0.1cm} \includegraphics[width=4.5cm]{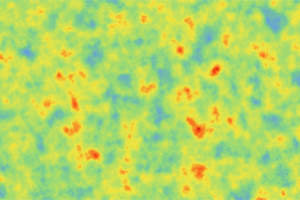}
    \end{center}
    \caption{Simulation of a Gneiting model associated with an exponential spatial covariance and a power variogram 
    $\gamma(u)$.} 
    \label{fig:pow2}
\end{figure}

\medskip

\section{Discussion}
\label{sec:discussion}

\subsection{Experimental reproduction of the spatio-temporal structure}
\label{sec:fluctuations}
The reproduction of the covariance structure can be experimentally validated by calculating the sample variograms of a set of realizations and comparing them with the theoretical variogram $\gamma(\bh,u)=1-C(\bh,u)$. An example is shown on Fig. \ref{fig:fluctuations}, where fifty realizations have been generated on a field in $\R^2 \times \R$ with $100 \times 100 \times 100$ nodes with a spatial mesh of $1 \times 1$ and a temporal mesh of $0.2$, with the same covariance model as in Fig \ref{fig:gtg} and $p=5000$ in both simulation approaches. Three spatial and three temporal variograms have been calculated, for time lags $u=0$, $u=0.2$ and $u=1.6$ and space lags $\bh=(0,0)$, $\bh=(6,6)$ and $\bh=(10,10)$, respectively. In all  cases, the experimental variograms fluctuate without any bias around the expected model, their average over the realizations matching almost perfectly the theoretical variogram, which corroborates the correctness of the proposed algorithms. Not surprisingly, the spectral approach provides experimental variograms that exhibit, for the same number of basic random fields ($p=5000$) in the sums \eqref{eq:spectral00} and \eqref{eq:substitution3}, slightly higher fluctuations than the substitution approach \citep{lantu1994}. Also note the dimple (hole effect) of the temporal variogram associated with the spatial lag $\bh=(10,10)$, a well-known property of the Gneiting model that arises even when the function $\psi$ is monotonic \citep{Kent,cuevas2017}.

\medskip

\begin{figure}[htp]
    \begin{center}
    \includegraphics[width=7.5cm]{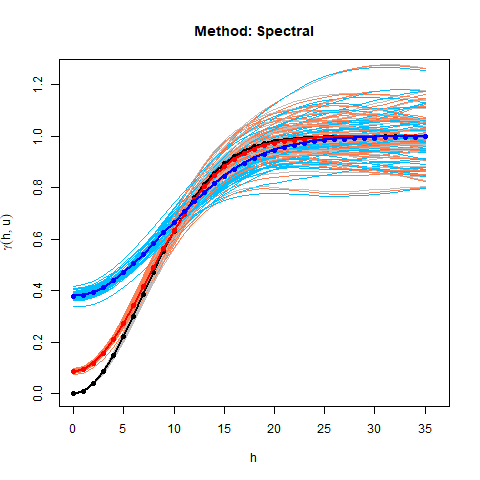} \hspace{0.1cm} \includegraphics[width=7.5cm]{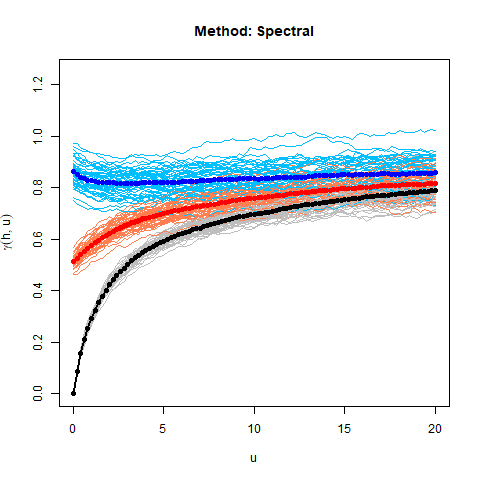} \hspace{0.1cm}     \end{center}
    \vspace{-0.4cm}
    \begin{center}
    \includegraphics[width=7.5cm]{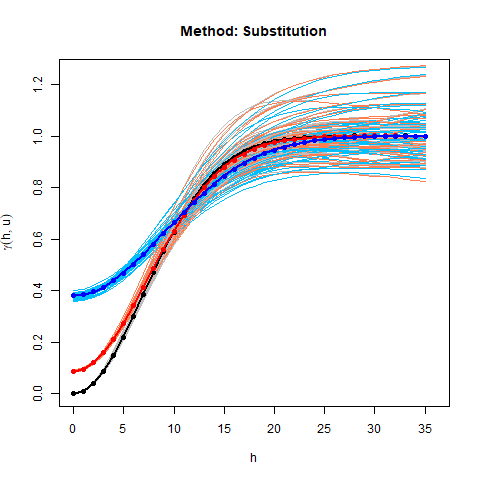} \hspace{0.1cm} \includegraphics[width=7.5cm]{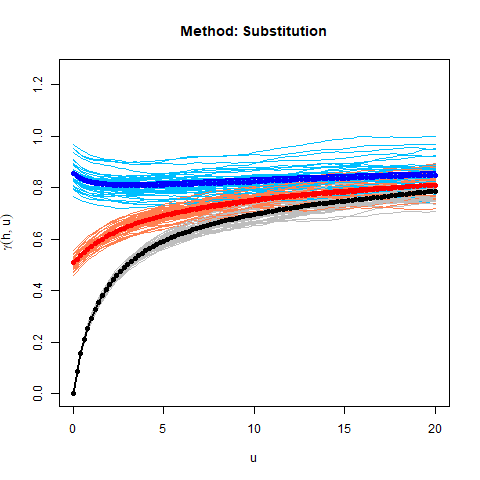} \hspace{0.1cm}     \end{center}
    \caption{Experimental spatial and temporal variograms (solid thin lines) for fifty realizations of the model in Fig \ref{fig:gtg} obtained with the spectral approach (top) and the substitution approach (bottom) on a $100 \times 100 \times 100$ domain of $\R^2 \times \R$ with spatial mesh $1 \times 1$ and temporal mesh $0.2$. Three spatial variograms are drawn in the left column, associated with $u=0$ (black), $u=0.2$ (red) and $u=1.6$ (blue). Three temporal variograms are drawn in the right column, associated with $\bh=(0,0)$ (black), $\bh=(6,6)$ (red) and $\bh=(10,10)$ (blue). In each case, the mean of the experimental variogram (dots) and the theoretical variograms (solid thick lines) are superimposed.} 
    \label{fig:fluctuations}
\end{figure}

\medskip

\subsection{Comparison of the simulation approaches}
\label{sec:approaches}
We presented two approaches for simulating spatio-temporal random fields with a Gneiting-type covariance function, based on two ingredients: a scale mixture argument for the spatial structure and the use of variograms for the temporal structure.  Despite algorithmic differences, these two approaches are mathematically very close in the sense that, conditional on the spatial scale $r$, they rely on the decomposition of the spectral measure
$$ F (d\bomega,d\tau \mid r ) = F^{\phantom{|}}_S (d\bomega \mid r) \, F^{\phantom{|}}_T (d\tau \mid \bomega , r).$$
Sampling from the marginal spatial spectral measure is common to both approaches, as the random vector $\sqrt{2 R} \, \tilde{\bOmega}$ used in the substitution approach \eqref{eq:substitution1} has the same distribution ${\cal N}_k(\bzero,2 R \, \bI_k)$ as the random vector $\bOmega(R)$ used in the spectral approach  \eqref{eq:ftspl}. The two methods only differ in the way of handling the temporal dimension.

The spectral approach uses the spectral measure associated with the temporal variogram $\gamma(u)$  in order 
to sample a conditional temporal frequency. The spectral measure is known for several classes of variogram functions, see  
Table \ref{tab:vario}. However, it is not known for some popular temporal structures such as
$\gamma(u) = (a|u|^{2\alpha} + 1)^\beta-1$. This function was used in 
\cite{Gneiting2002} and in many other studies implying climate variables, see 
e.g. \cite{bourotte2016flexible}. In this case however, a specific algorithm was designed, see example 3 above in Section \ref{sec:spectral}. A strong advantage of the spectral approach is that it is continuous in $\R^k \times 
\R$ and requires limited storage space since it only uses $p$ independent copies of random vectors of length $k+4$. 
The random field can then be computed at any location $(\bx,t)$ using~\eqref{eq:spectral00}. 

The spectral method can be extended to the simulation in $\R^k \times \R^l$, as it is relies on Theorem \ref{thm:cond} 
that can easily be generalized to that space.  In this setting,  the simulation of the spatial density remains unchanged; conditional upon $(\bOmega,R)$, the vector $\bTau=(\Tau_1,...,\Tau_l)$ is still an infinitely divisible 
random vector. The only difficulty is the simulation of multivariate infinitely 
divisible distributions. But, as it was the case for $l=1$, a case by case approach can be considered.

\smallskip

In the substitution approach, one simulates a pure spatial drift and a temporal intrinsic random field $W(t)$ with Gaussian increments and project them onto $(\bx,t)$ using a cosine function, as in \eqref{eq:substitution1}.
One advantage of using a standardized Gaussian random vector $\tilde{\bOmega}$ instead of a vector $\bOmega(R)$ with components of variance $2R$ is the possibility to apply the algorithm even if $R$ is zero, which happens with a non-zero probability if the measure $\mu$ defined in \eqref{eq:Laplace} has an atom at $0$. Although this case has been excluded in the presentation of both simulation approaches, it would not imply any change in the proposal in Section \ref{sec:substitution} and the demonstration in Appendix B. Another advantage of the substitution approach is that it is very generic, in the sense that the spectral measure of $\gamma(u)$ does not need to be known. Several discrete simulation algorithms such as the covariance matrix decomposition, the discrete spectral or the Gibbs propagation are possible in order to simulate $W(t)$, as already pointed out in Section \ref{sec:spectral}. These algorithms are applicable not only in $\R$, but also in $\R^l$ with $l>1$, which makes straightforward the extension of the presented approach to $\R^k \times \R^l$. The only difference lies in that the intrinsic random field $W$ is now defined on $\R^l$, and so is its variogram $\gamma(\bu)$.
Simulating a random field $W$ on $\R^l$ is, however, more difficult than on $\R$, essentially because the number of points targeted for simulation usually increases with $l$.
Note also that, in this approach, $\gamma$ can be any variogram on $\R^l$, which proves that this is a sufficient condition for the Gneiting covariance in \eqref{eq:G2} to be a valid model. The assumption of continuity of $\gamma$ is not even needed: variograms with a nugget effect, which do not have a spectral representation, can be considered in the construction by substitution.

The disadvantage is that the random field will be simulated only at those temporal coordinates where $W(t)$ has been simulated. If $N_T$ denotes the number of such temporal coordinates, then the simulation requires $N_T + p(k+3)$ values to be stored. Although this is uncommon in practice, a huge number of temporal coordinates (say, $N_T > 10^6$) may make the computational requirements prohibitive for the aforementioned discrete algorithms.

\section{Conclusions and Perspectives}
\label{sec:conclusions}

Two algorithms have been presented to simulate space-time random fields with nonseparable covariance belonging to the Gneiting class. The first one relies on a spectral decomposition of the covariance and constructs the simulated random field as a weighted sum of cosine waves with random frequencies and phases. In the second algorithm, an intrinsic time-dependent random field is substituted for the temporal frequency, yielding another representation of the simulated random field as a mixture of cosine waves. The proposed algorithms have been tested and validated through synthetic case studies. Their computational requirements are affordable in terms of both memory storage and CPU time; in particular, the number of needed floating point operations is proportional to the number of target space-time locations. Also, the algorithms can be adapted to the simulation of random fields in $\R^k \times \R^l$ with $l>1$.

This work paves the road to many possible extensions. Rather straightforward extensions include the simulation of multivariate space-time random fields based on the Gneiting class such as those proposed in \cite{bourotte2016flexible} and the simulation of random fields on spheres cross time with covariance models similar to the Gneiting class, but involving Stieltjes functions instead of Bernstein functions, see \cite{whitePorcu2019completepicture}. Other nonseparable models could also be simulated using at least one of the approaches presented here, including models proposed in \cite{ma2003families}. Spatio-temporal random fields derived from SPDEs are characterized through their spectral measures, see  \cite{carrizo018general} for a general presentation of these models. Simulating such fields could in some cases be performed using our spectral approach under the condition that one is able to simulate from the spectral measure, which requires further scrutiny. An interesting feature of these models is that, contrarily to the Gneiting class, they are not necessarily fully symmetric and that the marginal covariance function $C_S$ can include non-monotonic behavior.

\section*{Acknowledgements}

Denis Allard and Christian Lantuéjoul acknowledge support of the RESSTE network funded by the Applied Mathematics and Informatics division of INRA. Xavier Emery acknowledges the support of grant CONICYT PIA AFB180004 (AMTC) from the Chilean Commission for Scientific and Technological Research.

%
\appendix
\section{Proofs for the spectral approach}

\subsection{Proof of Theorem \ref{thm=dens-spec}}

Substituting \eqref{eq:gmasr} in \eqref{eq:fttpl}, one obtains
\begin{eqnarray}
\E \left[ e^{ \textstyle i \, u \, \Tau(\bomega,r) } \right] & =  & \exp \bigl( - \lambda(\bomega,r) \gamma(u) \bigr) \nonumber\\
& = & \exp \left( - \lambda(\bomega,r)
\int_{\R}\left(1 - \cos( u x) \right)  \sm (dx)\right) \nonumber \\
& = & 
\exp \left(\lambda(\bomega,r)
\int_{\R}\left(e^{i u x}-1-\frac{i u x}{1+x^2}\right)  \sm(dx)\right), \quad u \in \R 
\label{eq:loiid}
\end{eqnarray}
where $\lambda(\bomega,r)=\frac{\lvert \bomega \rvert^2}{4r}$ and where the last equality stems from the symmetry of $\sm$ and the integrability condition \eqref{eq:bndsm}. 
Therefore,  the distribution $F^{\phantom{|}}_T ( d \tau \mid \bomega, r)$ is an infinitely divisible distribution with Lévy measure ${\nu}_{\bomega,r}=\lambda(\bomega,r)\sm$. 

\medskip

Assume first that $\sm(\R)=+\infty$.  In this case,  ${\nu}_{\bomega,r}(\R)=+\infty$ for $\bomega\ne 0$ and the Lévy measure ${\nu}_{\bomega,r}$ is absolutely continuous, since $\sm$ is absolutely continuous. Then, by applying Lemma 1 of \cite{sato82} the infinitely divisible distribution $F^{\phantom{|}}_T ( d \tau \mid \bomega, r)$ is absolutely continuous for $\bomega\ne 0$ and for a.e. $\bomega$. As a consequence, since $F^{\phantom{|}}_S (  d\bomega\mid r)$ is also absolutely continuous, it follows that
$$
F(d\bomega,d\tau)=\int_0^{+\infty}  F^{\phantom{|}}_S (  d\bomega\mid r)F^{\phantom{|}}_T ( d \tau \mid \bomega, r)\mu (dr)
$$
is absolutely continuous.

\medskip

Let us now assume $\vartheta =\sm(\R)<+\infty$. In this case, \eqref{eq:loiid} can be rewritten as 
$$
\E \left[ e^{ \textstyle i \, u \, \Tau(\bomega,r) } \right] = \exp \left( \lambda(\bomega,r)
\int_{\R}\left(e^{i u x}-1\right)  \sm(dx)\right), \quad u \in \R 
$$
by  symmetry of the finite measure $\sm$. Since $0 < \vartheta < +\infty$,  $\Tau(\bomega,r)$ has a compound Poisson distribution with
$\mathbb{P}\bigl(\Tau(\bomega,r)=0)\bigr)=\exp\bigl(-\vartheta\lambda(\bomega,r)\bigr)$, 
which implies that 
$$
\mathbb{P}\left(\Tau=0\right)=\E\left[\exp\left(-\frac{\vartheta\vert \bOmega\vert^2 }{4R}\right)\right] >0.
$$
In conclusion, when $\vartheta < \infty$, the distribution of $\Tau$ has an atom in 0. Hence, $F$ is not absolutely continuous in this case.

\subsection{Proof of Theorem \ref{thm:cond}}
Recall that 
\begin{equation}\label{eq:cst}
C (\bh, u \mid r ) = \frac{1}{\bigl(\gamma(u)+1 \bigr)^{k/2}}\, \exp \left( - \frac{r \lvert \bh \rvert^2} {\gamma(u)+1}  \right)  
\equiv \int_{\R^k} \int_\R  e^{ \textstyle i \langle h , \bomega \rangle + i u \tau } \,  F^{\phantom{|}}_T ( d \tau 
\mid \bomega , r) \, F^{\phantom{|}}_S ( d \bomega \mid r) .
\end{equation}
If $u=0$, then 
$$ C (\bh, 0 \mid r ) = \exp \left( - r \lvert \bh \rvert^2  \right) \equiv \int_{\R^k}  e^{ \textstyle i \langle h , 
\bomega \rangle} \, F^{\phantom{|}}_S ( d \bomega \mid r), $$
which is nothing but \eqref{eq:ftspl}. This shows that $ F^{\phantom{|}}_S$ possesses the density 
\begin{equation}\label{eq:spldsy}
f^{\phantom{|}}_S(\bomega \mid r ) = \frac{1}{(4\pi r)^{k/2}} \exp \left( - \frac{ \lvert \bomega \rvert^2}{4 r} \right).
\end{equation}
Plugging \eqref{eq:spldsy} into \eqref{eq:cst}, one obtains
\begin{equation}\label{eq:cmp1}
C (\bh, u \mid r ) = \frac{1}{(4\pi r)^{k/2}} \int_{\R^k} e^{ \textstyle i \langle h , \bomega \rangle} \, \exp 
\left( - \frac{ \lvert \bomega \rvert^2}{4 r} \right) \, \int_\R  e^{ \textstyle i u \tau } \,  F^{\phantom{|}}_T 
( d \tau \mid \bomega , r) \, d \bomega .
\end{equation}
On the other hand, up to a multiplicative factor, the function $ \bomega \, \mapsto \exp\left(-\, r \lvert \bh \rvert^2 / 
(\gamma(u)+1) \right)$ is the Fourier transform of a Gaussian random vector:
\begin{equation}\label{eq:cmp2}
C (\bh, u \mid r ) = \frac{1}{(4\pi r)^{k/2}} \int_{\R^k} e^{ \textstyle i \langle h , \bomega \rangle} \, \exp 
\left( - \frac{ \lvert \bomega \rvert^2 \, (\gamma(u)+1)}{4 r} \right) \,  d \bomega.
\end{equation}
Comparing \eqref{eq:cmp1} and \eqref{eq:cmp2}, the injectivity of the Fourier transform implies 
$$ \exp \left( - \frac{ \lvert \bomega \rvert^2}{4 r} \right) \, \int_\R  e^{ \textstyle i u \tau } \,  
F^{\phantom{|}}_T ( d \tau \mid \bomega, r) = \exp \left( - \frac{ \lvert \bomega \rvert^2 \, (\gamma(u)+1) }{4 r} 
\right) \qquad \bomega\text{-a.e.} $$
or equivalently 
$$ \int_\R  e^{ \textstyle i u \tau } \,  F^{\phantom{|}}_T ( d \tau \mid \bomega, r) = \exp \left( - \frac{ \lvert 
\bomega \rvert^2 \gamma(u)}{4 r} \right) \qquad \bomega\text{-a.e.},$$
which is precisely \eqref{eq:fttpl}. \hfill $\Box$

\subsection{On the generic approach}
The aim of this section is to show that $\sum_{n \geq 1} X_{\Tau_n}$ and $\Tau (\bomega,r)$ have the same distribution. 
This is done by comparing their Fourier transforms. Remind that the spectral measure of the variogram is positive, symmetric, without an atom at the origin, and satisfies the integrability property 
\begin{equation}\label{eq:bsm}
\int_{\R} \frac{x^2 \, \sm (dx)}{1 + x^2} = A < +\infty. 
\end{equation}
Let us start with
$$ \sm (dx) = \int_{\R_+} \exp \left( - t \frac{x^2}{1 + x^2} \right) \, \frac{x^2 \, \sm (dx)}{1 + x^2} \, dt. $$
Because of \eqref{eq:bsm}, the positive function $\theta$ defined on $\R_+$ by 
$$ \theta (t) = \int_{\R} \exp \left( - t \frac{x^2}{1 + x^2} \right) \, \frac{x^2 \, \sm (dx)}{1 + x^2} $$
is upper bounded by $A$. It follows that, for each $t>0$, the measure 
$$ \sm_t (dx) = \frac{1}{\theta(t)} \exp \left( - t \frac{x^2}{1 + x^2} \right) \, \frac{x^2 \, \sm (dx)}{1 + x^2} $$
is a probability measure on $\R$. This measure is symmetric, and satisfies
\begin{equation}\label{eq:bdn}
\sm (dx) = \int_{\R_+} \sm_t (dx) \, \theta (t) \, dt. 
\end{equation}

\medskip

Consider now a Poisson point process $\bigl(\Tau_n, n \geq 1 \bigr)$ with intensity $\lambda (t) = \lambda \, \theta(t)$ 
on $\R_+$ ($\lambda$ is put here as a short notation for $\frac{\lvert \bomega \rvert^2}{4r}$). Since $\theta$ $\lambda(t)$ is upper 
bounded by $\lambda A$, this process has no accumulation point. Consider also a family $\bigl( X_t, t \in \R_+ \bigr)$ 
of independent random variables distributed as $\sm_t$. Because $\sm_t$ is symmetric, 
the Fourier transform of $U_t$ can be written as
\begin{equation}\label{eq:frr}
\E \bigl[ \exp ( i u X_t ) \bigr] = \int_{\R} \cos (u x) \, \sm_t (dx). 
\end{equation}
In what follows, we calculate the Fourier transform of $\Tau = \sum_{n \geq 1} X_{\Tau_n}$. Denoting by 
$ \Lambda (t_0)$ the integral of $\lambda (t)$ on $]0,t_0[$, we have
\begin{align*}
\E \bigl[ \exp ( i u \Tau ) \bigr] 
&= \lim_{t_0 \longrightarrow \infty} \sum_{n=0}^\infty \exp \bigl( - \Lambda (t_0) \bigr) \, \frac{\Lambda^n 
(t_0)}{n!} \, \left[ \int_0^{t_0} \frac{\lambda(t)}{\Lambda (t_0)} \, \E \bigl[ \exp ( i u X_t ) \bigr] \, dt 
\right]^n \\
&= \lim_{t_0 \longrightarrow \infty} \exp \left( \int_0^{t_0} \E \bigl[ \exp ( i u X_t ) - 1 \bigr] \, \lambda(t) 
\, dt \right) \\
&= \exp \left( \int_0^{\infty} \E \bigl[ \exp ( i u X_t ) - 1 \bigr] \, \lambda(t) \, dt \right).  
\end{align*}
This implies, owing to \eqref{eq:frr}
$$ \E \bigl[ \exp ( i u \Tau ) \bigr] = \exp \left( \int_0^{\infty} \int_{\R} \bigl[ \cos (u x) - 1 \bigr] \, \sm_t (dx) \, 
\lambda(t) \, dt \right). $$
Permuting the integrals and replacing $\lambda (t)$ by its expression, we obtain 
$$ \E \bigl[ \exp ( i u \Tau ) \bigr] = \exp \left( \frac{ \vert \bomega \rvert^2}{4 r} \, \int_{\R} \bigl[ \cos (u x) 
- 1 \bigr] \, \sm (dx) \right). $$
Finally, the spectral representation \eqref{eq:gmasr} of $\gamma$ gives
$$ \E \bigl[ \exp ( i u \Tau ) \bigr] = \exp \left( - \frac{ \vert \bomega \rvert^2}{4 r} \gamma (u) \right) , $$
which is precisely the Fourier transform \eqref{eq:fttpl} of $\Tau (\bomega,r)$.  \hfill $\Box$

\bigskip

\subsection{Implementing the generic approach for logarithmic variograms}
The construction of $\theta$ and $\sm_t$ proposed in appendix A.3 is not necessarily unique. Starting from 
$$ \sm (d x) = \frac{\exp (- a \lvert x \rvert)}{ \lvert x \rvert \, \ln a^2} = \frac{1}{\ln a^2} \int_a^{+\infty} 
\exp \bigl(- t \, \lvert x \rvert \bigr) \, d t , $$
it appears that a possible decomposition such as \eqref{eq:dcn} can be obtained by taking 
$$ \theta (t) = \frac{2}{t \, \ln a^2} \, 1^{\phantom{|}}_{t \geq a} = \frac{1}{t \, \ln a} \, 
1^{\phantom{|}}_{t \geq a} $$
and 
$$ \sm_t (dx) = \frac{1}{2}  t \, e^{ - t \lvert x \vert} \, dx, \qquad t > a .$$

Consider now a Poisson point process $ \bigl( \Tau_n , \, n \geq 1 \bigr) $ on $\R_+$ with intensity function
$$ \lambda (t) = \frac{ \lvert \bomega \rvert^2}{4 \, r } \, \theta (t) \equiv \frac{\lambda}{t} \, 
1^{\phantom{|}}_{t \geq a}. $$
A simple approach to simulate this point process is to take the inverse of a homogeneous point process through the primitive $ \Lambda (t) = \lambda \, \ln ( t / a) $ of the intensity function that vanishes at $a$, 
as shown in Fig. \ref{fig:ans}. In this figure, the $U_i$'s are independent standard uniform variables and are 
related to the Poisson times by the formula $ \lambda \, \ln ( \Tau_n / a ) = - \ln \bigl( U_1 \cdots U_n \bigr)$, 
which gives
\begin{equation}\label{eq:tn}
\Tau_n = \frac{a}{(U_1 \cdots U_n )^{1 / \lambda}} . 
\end{equation}

\medskip

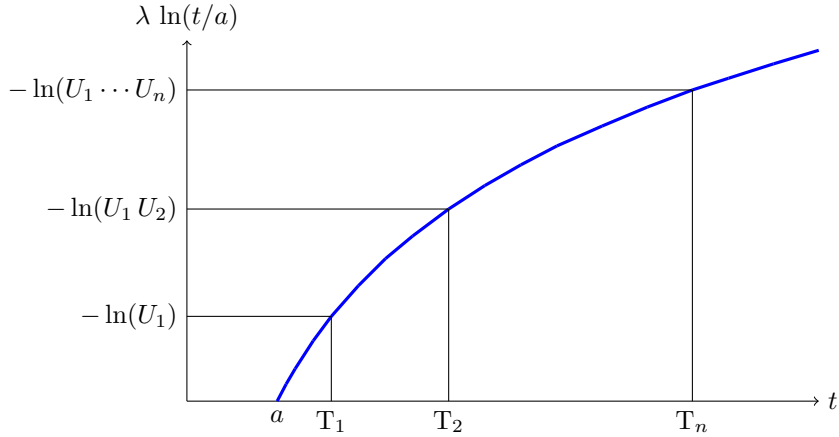
\begin{figure}[htb]
\centering
\begin{tikzpicture}[scale=1.2]
\draw[->] (0,0) -- (7,0) ;
\draw (7,0) node[right] {$t$};
\draw[->] (0,0) -- (0,4) ;
\draw (0,4) node[above] {$ \lambda \, \ln( t / a)$};
\draw [blue,very thick] (1.0,0.00) -- (1.1,0.19) ;
\draw [blue,very thick] (1.1,0.19) -- (1.2,0.36) ;
\draw [blue,very thick] (1.2,0.36) -- (1.4,0.67) ;
\draw [blue,very thick] (1.4,0.67) -- (1.6,0.94) ;
\draw [blue,very thick] (1.6,0.94) -- (1.9,1.28) ;
\draw [blue,very thick] (1.9,1.28) -- (2.2,1.58) ;
\draw [blue,very thick] (2.2,1.58) -- (2.5,1.83) ;
\draw [blue,very thick] (2.5,1.83) -- (2.9,2.13) ;
\draw [blue,very thick] (2.9,2.13) -- (3.3,2.39) ;
\draw [blue,very thick] (3.3,2.39) -- (3.7,2.62) ;
\draw [blue,very thick] (3.7,2.62) -- (4.1,2.83) ;
\draw [blue,very thick] (4.1,2.83) -- (4.6,3.05) ;
\draw [blue,very thick] (4.6,3.05) -- (5.1,3.26) ;
\draw [blue,very thick] (5.1,3.26) -- (5.6,3.45) ;
\draw [blue,very thick] (5.6,3.45) -- (6.0,3.58) ;
\draw [blue,very thick] (6.0,3.58) -- (6.5,3.74) ;
\draw [blue,very thick] (6.5,3.74) -- (7.0,3.89) ;
\draw (1,0) node[below] {$a$} ;
\draw (0,0.94) node[left] {$- \ln(U_1)$};
\draw (0,0.94) -- (1.6,0.94) ;
\draw (1.6,0.94) -- (1.6,0) ;
\draw (1.6,0) node[below] {$\Tau_1$} ;
\draw (0,2.13) node[left] {$- \ln(U_1 \, U_2)$};
\draw (0,2.13) -- (2.9,2.13) ;
\draw (2.9,2.13) -- (2.9,0) ;
\draw (2.9,0) node[below] {$\Tau_2$} ;
\draw (0,3.45) node[left] {$- \ln(U_1 \cdots U_n)$};
\draw (0,3.45) -- (5.6,3.45) ;
\draw (5.6,3.45) -- (5.6,0) ;
\draw (5.6,0) node[below] {$\Tau_n$} ;
\end{tikzpicture}
\caption{Simulation of a heterogeneous Poisson point process}\label{fig:ans}
\end{figure} 

\medskip

Now, recall that $\Tau$ has the same distribution as $\sum_{n=1}^{+\infty} X_{\Tau_n}$, where each $X_t$ is distributed 
as $\nu_t$. Because the $X_t$'s are independent, we have
$$ \mathrm{Var} \Bigl[ \sum_{n \geq 1} X_{\Tau_n} \Bigr] =  \sum_{n \geq 1} \mathrm{Var} \bigl[ X_{\Tau_n} \bigr]. $$
Moreover, \eqref{eq:tn} implies
$$ \mathrm{Var} \bigl[ X_{\Tau_n} \bigr] = \E \Bigl[ \mathrm{Var} \bigl[ X_{\Tau_n} \vert \Tau_n \bigr] \Bigr] = \E \Bigl[ \frac{2}{\Tau^2_n} \Bigr] 
= \frac{2}{a^2} \, \E \Bigl[ (U_1 \cdots U_n)^{2 \lambda} \Bigr] = \frac{2}{a^2} \, \left( \frac{\lambda}{\lambda + 2} 
\right)^n. $$
Consequently
$$ \mathrm{Var} \Bigl[ \sum_{n \geq 1} X_{\Tau_n} \Bigr] = \frac{2}{a^2} \, \sum_{n \geq 1} \left( \frac{\lambda}{\lambda + 2} 
\right)^n = \frac{\lambda}{a^2}. $$
Similarly, if the series is truncated at order $n_0$, then the same calculation leads to the residual variance  
$$ \mathrm{Var} \Bigl[ \sum_{n \geq n_0+1} X_{\Tau_n} \Bigr] = \frac{2}{a^2} \sum_{n \geq n_0+1} \left( \frac{\lambda}{\lambda + 2} 
\right)^n = \left( \frac{\lambda}{\lambda + 2} \right)^{n_0} \frac{\lambda}{a^2}. $$
Let $\varepsilon > 0$ be arbitrarily small. From the previous calculations, it follows that  
$$ \frac{ \mathrm{Var} \Bigl[ \sum_{n \geq n_0+1} X_{\Tau_n} \Bigr]}{\mathrm{Var} \Bigl[ \sum_{n \geq 1} X_{\Tau_n} \Bigr] } < \varepsilon \ \Longleftrightarrow \ \left( \frac{\lambda}{\lambda + 2} \right)^{n_0} < 
\varepsilon \ \Longleftrightarrow \  n_0 > \frac{- \ln \epsilon}{\ln ( 1 + 2 / \lambda) }. $$
\hfill $\Box$

\section{Proofs for the substitution approach}

\subsection{Proof of Theorem \ref{theorem2}}

For $(\bx,t) \in \R^k \times \R$, $Z(\bx,t)$ conditional on $(R,\tilde{\bOmega},V,W)$ (i.e., only letting $\Phi$ vary randomly) has a zero expectation, insofar as it is proportional to the cosine of a random variable uniformly distributed on an interval of length $2 \pi$. The prior expectation of $Z(\bx,t)$ is therefore zero:
\begin{equation*}
\E [Z(\bx,\bt) ] = \E [ \E[ Z(\bx,t) \mid R, \tilde{\bOmega}, V, W ] ] = 0.
\end{equation*}
Let us now calculate the covariance between the random variables $Z(\bx,t)$ and $Z(\bx^{\prime},t^{\prime})$, with $(\bx,t) \in \R^k \times \R$ and $(\bx^{\prime},t^{\prime}) \in \R^k \times \R$:
\begin{equation*}
\begin{split}
\E & [ Z(\bx,t) Z(\bx^{\prime},t^{\prime})] \\
& = 2 \E \bigg[-\ln(U) \cos \left(\sqrt{2R} \, \langle \tilde{\bOmega}, \bx \rangle + \frac{\lvert \tilde{\bOmega}\rvert}{\sqrt{2}} W(t) + \Phi \right) \cos \left(\sqrt{2R} \, \langle \tilde{\bOmega}, \bx^{\prime} \rangle + \frac{\lvert \tilde{\bOmega}\rvert}{\sqrt{2}} W(t^{\prime}) + \Phi \right) \bigg] \\
& = 2 \E \bigg[\cos \left(\sqrt{2R} \, \langle \tilde{\bOmega},\bx \rangle + \frac{\lvert \tilde{\bOmega}\rvert}{\sqrt{2}} W(t) + \Phi \right) \cos \left(\sqrt{2R} \, \langle  \tilde{\bOmega}, \bx^{\prime} \rangle + \frac{\lvert \tilde{\bOmega}\rvert}{\sqrt{2}} W(t^{\prime}) + \Phi \right) \bigg].
\end{split}
\end{equation*}
The last equality stems from the fact that $-\ln(U)$ is an exponential random variable with mean $1$ and is independent of $(R,\tilde{\bOmega},W)$. Using the product-to-sum trigonometric identities, one can write the product of cosines as half the sum of two cosines, namely:
\begin{itemize}
    \item the cosine of the difference: $\cos \left(\sqrt{2R} \, \langle \tilde{\bOmega}, \bx-\bx^{\prime} \rangle + \frac{\lvert \tilde{\bOmega}\rvert}{\sqrt{2}} (W(t)-W(t^{\prime})) \right)$
    
    \item the cosine of the sum: $\cos \left(\sqrt{2R} \, \langle \tilde{\bOmega}, \bx+\bx^{\prime} \rangle + \frac{\lvert \tilde{\bOmega}\rvert}{\sqrt{2}} (W(t)+W(t^{\prime})) + 2 \Phi \right)$.
\end{itemize}
However, because $\Phi$ is uniformly distributed on $(0, 2\pi)$ and independent of $(R, \tilde{\bOmega},W)$, the expectation of the cosine of the sum is 0. It remains 
\begin{equation*}
\E [Z(\bx,t) Z(\bx^{\prime},t^{\prime})]
 = \E \bigg[\cos \left(\sqrt{2R} \, \langle \tilde{\bOmega}, \bx-\bx^{\prime} \rangle + \frac{\lvert \tilde{\bOmega}\rvert}{\sqrt{2}} (W(t)-W(t^{\prime})) \right)\bigg].
\end{equation*}\\
The increment $W(t)-W(t^{\prime})$ is a Gaussian random variable with zero mean and variance $2 \gamma(t-t^{\prime})$ and is independent of $(R,\tilde{\bOmega})$, i.e.:
\begin{equation*}
W(t)-W(t^{\prime}) = \sqrt{2 \gamma(t-t^{\prime})} Y, 
\end{equation*}
with $Y \sim {\cal N} (0,1)$ independent of $(R,\tilde{\bOmega})$. Defining $\bh = \bx - \bx^{\prime}$ and $u = t - t^{\prime}$ and denoting by $g$ the standard Gaussian probability density, one therefore obtains:
\begin{equation}
\begin{split}
\E & [ Z(\bx,t) Z(\bx^{\prime},t^{\prime})] \\
  & = \int_{\R_+} \int_{\R^k} \int_{\R} \cos \left(\sqrt{2r} \, \langle \tilde{\bomega}, \bh \rangle + \lvert \tilde{\bomega} \rvert \sqrt{\gamma(u)} y \right) g(y) dy  \frac{1}{(2\pi)^{k/2}} \, \exp \left( -\frac{\lvert \tilde{\bomega} \rvert^2}{2} \right) d \tilde{\bomega} \mu(dr) \\ 
  & = \frac{1}{(2\pi)^{k/2}} \int_{\R_+} \int_{\R^k} \cos \left(\sqrt{2r} \, \langle \tilde{\bomega}, \bh \rangle \right) \int_{\R} \cos \left(\lvert\tilde{\bomega}\rvert \sqrt{\gamma(u)} y \right) g(y) dy \, \exp \left( -\frac{\lvert \tilde{\bomega} \rvert^2}{2} \right) d\tilde{\bomega} \mu(dr) . 
\end{split}
\label{eq:emery2}
\end{equation}
The last equality in \eqref{eq:emery2} stems from the angle-sum trigonometric identity and the fact that $\langle \tilde{\bomega}, \bh \rangle$ is an odd function of $\tilde{\bomega}$ and $\lvert\tilde{\bomega}\rvert \sqrt{\gamma(u)} y$ is an even function of $\tilde{\bomega}$.

The simulated random field $Z$ is therefore second-order stationary, since its expectation is identically zero and the covariance between any two variables $Z(\bx,t)$ and $Z(\bx^{\prime},t^{\prime})$ only depends on $\bh = \bx - \bx^{\prime}$ and $u = t - t^{\prime}$. Up to a multiplicative factor, the last integral in \eqref{eq:emery2} appears as the Fourier transform of the standard Gaussian probability density $g(y)$ on $\R$. Specifically:
\begin{equation*}
\int_{\R} \cos \left(\lvert\tilde{\bomega}\rvert \sqrt{\gamma(u)} y \right) g(y) dy 
= \exp \left( -\lvert\tilde{\bomega}\rvert^2 \frac{\gamma(u)}{2} \right). 
\end{equation*}
Hence:
\begin{equation}
\begin{split}
\E [ Z(\bx,t) Z(\bx^{\prime},t^{\prime})]
  & = \frac{1}{(2\pi)^{k/2}} \int_{\R_+}  \int_{\R^k}  \cos(\sqrt{2r} \, \langle \tilde{\bomega},\bh \rangle) \exp \left(- \lvert\tilde{\bomega}\rvert^2 \frac{\gamma(u)+1}{2} \right) d\tilde{\bomega} \mu(dr) \\ 
  & = C(\bh,u).
\end{split}
\label{eq:emery3}
\end{equation}
The last equality in \eqref{eq:emery3} stems from \eqref{eq:emery111} and completes the proof. \hfill $\Box$

\clearpage

\bibliographystyle{apalike}

\end{document}